\documentclass{amsart}

\def\baselinestretch{1}

\usepackage[utf8]{inputenc}
\usepackage{subcaption} 
\usepackage{caption}
\usepackage{epsfig}
\usepackage{etoolbox}
\usepackage{packagearticleieee}
\usepackage{packagevariables} 
\usepackage{cite}
\usepackage[hyphens]{url}
\usepackage{hyperref}
\hypersetup{
    breaklinks = true
}
\makeatletter
\let\NAT@parse\undefined
\makeatother

\begin{document}

\begin{abstract}
Ensuring safety in cyber-physical systems (CPSs) is a critical challenge, especially when system models are difficult to obtain or cannot be fully trusted due to uncertainty, modeling errors, or environmental disturbances. Traditional model-based approaches rely on precise system dynamics, which may not be available in real-world scenarios. To address this, we propose a data-driven safety verification framework that leverages matrix zonotopes and barrier certificates to verify system safety directly from noisy data. Instead of trusting a single unreliable model, we construct a set of models that capture all possible system dynamics that align with the observed data, ensuring that the true system model is always contained within this set. This model set is compactly represented using matrix zonotopes, enabling efficient computation and propagation of uncertainty. By integrating this representation into a barrier certificate framework, we establish rigorous safety guarantees without requiring an explicit system model. Numerical experiments demonstrate the effectiveness of our approach in verifying safety for dynamical systems with unknown models, showcasing its potential for real-world CPS applications.
\end{abstract}

\title[Data-Driven Safety Verification using Barrier Certificates and Matrix Zonotopes]{Data-Driven Safety Verification using Barrier Certificates \\ and Matrix Zonotopes}

\author[mohammed adib oumer]{Mohammed Adib Oumer$^1$} 
\author[amr alanwar]{Amr Alanwar$^2$} 
\author[majid zamani]{Majid Zamani$^1$}

\address{$^1$Department of Computer Science at the University of Colorado, Boulder, CO, USA.}
\email{\{mohammed.oumer,~majid.zamani\}@colorado.edu}
\urladdr{https://www.hyconsys.com/members/moumer/}
\urladdr{ https://www.hyconsys.com/members/mzamani/}

\address{$^2$TUM School of
Computation, Information and Technology at the Technical University of Munich, Heilbronn, Germany.}
\email{alanwar@tum.de}
\urladdr{https://www.professoren.tum.de/en/alanwar-amr}

\maketitle

\section{Introduction}
 
Ensuring safety is a fundamental requirement in the design and operation of cyber-physical systems (CPSs), particularly in safety-critical applications such as autonomous vehicles, robotics, power grids, and industrial automation~\cite{Althoff2010PhD}. Barrier certificates have emerged as a powerful and widely adopted tool for formal safety verification in dynamical systems~\cite{prajna2004safety}. These certificates are Lyapunov-like functions that ensure that system trajectories remain within a predefined safe region and do not reach an unsafe region. 
A barrier certificate is a real-valued function that is nonpositive over the initial states, positive over the unsafe states, and nonincreasing with transitions. 
Thus, such a certificate guarantees safety as its zero-level set separates the reachable and unsafe states. 
Moreover, its zero-sublevel set over-approximates the set of reachable states.
By formulating safety conditions as inequalities involving the system’s dynamics, a barrier certificate provides a computationally tractable alternative to reachability analysis. However, traditional approaches~\cite{prajna2004safety,ames2019control,jagtap2020formal,wang2018permissive,prajna2007framework,prajna2004stochastic} assume explicit knowledge of the system model, significantly limiting their applicability in scenarios where the true system dynamics are unknown or uncertain.

To overcome this limitation, recent research has shifted towards data-driven approaches that construct barrier certificates directly from observed system behavior. The authors in~\cite{BCFundamentalLemma} introduced a method for synthesizing barrier certificates using the Fundamental Lemma, which allows for the derivation of system properties from input-output data. While this approach removes the need for an explicit system model, it inherently assumes noiseless data, as the Fundamental Lemma relies on exact trajectory observations to infer system behavior. Consequently, its applicability is limited in real-world scenarios, where measurement noise, disturbances, and modeling uncertainties are inevitable. The study in~\cite{salamati2022data,salamati2024data} explores data-driven safety verification of unknown discrete-time stochastic systems, presenting a novel approach that formulates the computation of barrier certificates as a robust convex optimization problem. This formulation leverages trajectory-based sampling to construct constraints, ensuring that the derived barrier certificates accurately capture system behavior under uncertainty. 
The paper in~\cite{schon2024data} transforms probabilistic barrier certificate constraints into a data-driven optimization framework by defining an ambiguous set of possible transition kernels. To tackle the resulting problem, it leverages sum-of-squares (SOS) optimization~\cite{SOS}. 

Recent advancements in data-driven reachability analysis have provided a promising alternative to traditional barrier certificate methods. A variety of approaches have emerged to tackle the problem of the unknown model from different perspectives. One approach in~\cite{Devonport2020} introduced two probabilistic methods: one reformulated reachability as a classification problem, while the other applied Monte Carlo sampling to estimate reachable sets. Similarly,~\cite{meas-driven} proposed a hybrid strategy combining measurement-driven and model-based formal verification techniques. Other works have focused on refining reachability estimates through active learning~\cite{conf:activelearning1,conf:activelearning2}. A probabilistic framework for general nonlinear systems, based on Christoffel function level sets, was developed in~\cite{conf:murat_christoffel}. More recently, neural network-based approaches have gained traction. In~\cite{lars2023Conformal}, feedforward neural networks~\cite{feedforwardNN} were employed to approximate reachable sets, with uncertainty quantified via conformal inference~\cite{Conformalbook2005,Conformal2014}. Additionally, kernel density estimation, accelerated by the fast Fourier transform, has been applied to model uncertainties and compute probabilistic reachable sets~\cite{kernelprobRA}. A distinct data-driven framework for forward stochastic reachability analysis was introduced in~\cite{BootstrappedDDSRA}. This method estimates the evolution of the state’s probability density function using trajectory data, which is then used to construct a Gaussian mixture model. Although these probabilistic approaches enhance estimation accuracy with larger datasets, they still lack robust safety guarantees. In particular, rare but safety-critical events remain challenging to capture, limiting the reliability of purely probabilistic reachability methods.

Conversely, researchers have also investigated robust overapproximations of reachability using noise-free data. In~\cite{conf:onthefly}, interval Taylor-based techniques were applied to systems modeled by differential inclusions, with later extensions incorporating noisy data. Another notable approach, proposed in~\cite{conf:tomlin_nonlin_reach}, integrated partial model knowledge with data-driven learning to capture state-dependent uncertainties, relying on the assumption that the unknown dynamics conformed to a known bounding set. Additionally,~\cite{LuciaBackward} explored the computation of robust backward reachable sets for unknown linear systems, leveraging noisy data to ensure safety guarantees despite uncertainty. The method proposed in~\cite{alanwar2021data,Alanwar2023Datadriven} formulates a set of possible system models encapsulated within a matrix zonotope, guaranteeing that every model consistent with the noisy data and prescribed noise bounds is included. However, this approach often results in conservative reachable sets for long horizons. This conservatism can be partially mitigated if there is additional side information available~\cite{alanwar2022enhancing}. A key approach to providing infinite horizon safety guarantees is through the use of barrier certificates.

In this paper, we build on the results in~\cite{alanwar2021data,Alanwar2023Datadriven} to compute a set of models for a linear control system and leverage barrier certificates as a discretization-free approach to provide safety guarantees over an infinite horizon for this given set of models. 
We use sum-of-squares (SOS) optimization to search for these certificates. To tackle the presence of a set of models, we start by formulating a computationally heavy condition for safety directly from the computed model parameters. 
Normally, the transition function is a function of the state, input, and noise parameters in some compact set. 
In the condition we formulate, the transition function is additionally a function of the model parameters that are also in some compact set. 
This dramatically increases the total number of parameters we need to keep track of, particularly when the barrier certificate is applied over the transition function. 
The larger memory usage makes it harder to run the optimization problem successfully. 
Thus, we address this computational challenge of the given condition by stating an equivalent condition that exploits the matrix zonotope representation of the models of the system and scales relatively well with the dimension of the system. 
We substantiate the benefit of the latter condition by implementing both conditions and investigating their performance for systems with different dimensions. 
As the models of the system can be made to be less conservative using available side information, our approach can also be integrated with the enhanced model parameters. 
Moreover, as the computation of matrix zonotope over-approximation of system model parameters has been generalized for polynomial nonlinear control systems, our approach seems promising to be adopted for such systems with appropriate updates. The codes to recreate our findings are publicly available.
\footnote{
[Online]. Available: \url{https://github.com/maoumer/Data-Driven-Safety-Verification-using-Barrier-Certificates-and-Matrix-Zonotopes}
}

The remainder of the paper is organized as follows. Section~\ref{sec:preliminaries} introduces preliminaries, notation, and problem formulation used in the paper. Section~\ref{sec:DDSV} covers our approach of safety verification via barrier certificates for the data-driven system models while section~\ref{sec:computation} highlights a computation method of finding the barrier certificates via SOS programming. Section~\ref{sec:numerical-simulations} illustrates numerical simulations. Finally, Section~\ref{sec:conclusion} offers conclusive remarks.

\section{Preliminaries and Problem Formulation}\label{sec:preliminaries}

We start with the notations used in the paper, followed by the preliminaries and problem formulations.

\subsection{Notations}

The set of natural numbers, nonnegative integers, positive reals, the real $n$-dimensional space, and the space of real $n \times m$ matrices are denoted by $\N$, $\Z_{\geq 0}$, $\R_{>0}$, $\R^n$, and $\R^{n \times m}$, respectively. For a matrix $A$, $(A)_j$ denotes its $j$th column, $A^{\top}$ represents its transpose, and $A^\dagger$ its Moore-Penrose pseudoinverse. The operator $\diag \cdot$ constructs a diagonal matrix from its arguments. The vectors $\1_n$ and $\zero_n$ are the $n$-dimensional vectors with all entries equal to $1$ and $0$, respectively. We use $\1$ and $\zero$ to represent matrices of appropriate dimensions with all entries equal to $1$ and $0$, respectively. If a system, denoted by $\mathcal{S}$, satisfies a property $\Psi$, it is written as $\mathcal{S} \models \Psi$. 
Given a variable $x$ that can have values in the range from $x_{min}$ to $x_{max}$, we use $[x_{min},x_{max}]$ to represent the interval.

\subsection{Set representation}
In this work, we utilize zonotopes as a mathematical representation of sets defined as follows.  
\begin{definition}[Zonotope~\cite{Kuhn1998}]
    Given a center $c_{\zon} \in \R^n $ and a number $\gamma_\zon\in \N$ of generator vectors $g_{\zon_i}\in\R^n$, a zonotope is defined as
    \begin{equation*}
        \zon = \Big\{ x \in \R^{n} \Big| x = c_{\zon} + \sum_{i=1}^{\gamma_\zon} \beta_i\, g_{\zon_i},-1 \leq \beta_i\leq 1 \Big\}.
    \end{equation*}
    We use the short notation $ \zon = \langle c_{\zon},G_{\zon} \rangle$ to denote a zonotope, where $G_{\zon} = [g_{\zon_1}\dots g_{\zon_{\gamma_\zon}}]\in \R^{n\times \gamma_\zon }$. 
\end{definition}

Given $L\in\R^{m\times n}$, the linear transformation of a zonotope $\zon$ is a zonotope $L\zon=\langle Lc_{\zon},LG_{\zon} \rangle$. Given two zonotopes $\zon_1 = \langle c_{\zon_1 },G_{\zon_1 } \rangle \subset \R^n$ and $\zon_2 = \langle c_{\zon_2 },G_{\zon_2 } \rangle\subset \R^n$, the Minkowski sum and the Cartesian product are respectively computed as~\cite{farjadnia2024robust}:
\allowdisplaybreaks
\begin{align*}
\zon_1 \oplus \zon_2  
    &= \Big\langle c_{\zon_1 }+ c_{\zon_1 },[G_{\zon_1 } \, \, G_{\zon_2 }] \Big\rangle,
    \\
    \zon_1 \times \zon_2 
    &= \left\langle \begin{bmatrix} c_{\zon_1}\\ 
    c_{\zon_2} \end{bmatrix},
    \begin{bmatrix} G_{\zon_1} & \zero \\ \zero & G_{\zon_2} \end{bmatrix}
   \right\rangle.
\end{align*}

A matrix zonotope is represented as $ \mzon = \langle C_{\mzon}, G_{\mzon} \rangle$, and its definition is analogous to that of a zonotope.  It is characterized by a center matrix $C_{\mzon} \in \R^{n\times m}$  and $\gamma_{\mzon} \in \N$ generator matrices $G_{\mzon} = \begin{bmatrix}G_{\mzon_1}& \cdots & G_{\mzon_{\gamma_\mzon}} \end{bmatrix}\in \R^{n\times (m\gamma_{\mzon}) }$~\cite[p.~52]{Althoff2010PhD}. 

Given a matrix zonotope $\mzon$, one can convert it to an interval matrix $\mzon_{int} = [M^{min}, M^{max}]$ where $M^{min}$ and $M^{max}$ are matrices and for $Z \in \mzon$, each component satisfies $M_{i,j}^{min} \leq Z_{i,j} \leq M_{i,j}^{max}$. This inequality can then be represented as a vector of polynomial inequalities over all elements of the matrix given by: 
\begin{align}
\label{eq:poly_inequality}
    \begin{bmatrix}
    Z_{i,j} - M_{i,j}^{min}\\
    M_{i,j}^{max} - Z_{i,j}\\
\end{bmatrix} \geq 0.
\end{align}
Note the comma to distinguish these interval matrices from the concatenation of matrices. A similar argument applies to zonotopes.

\subsection{System Dynamics}\label{sec:sysDyn}

We consider an unknown discrete-time linear  system $\mathcal{S}$ as follows:
\begin{equation}
    x_{k+1} = A_{\text{tr}} x_k + B_{\text{tr}} u_k + w_k,
\label{eq:model-linear}
\end{equation}
where $A_{\text{tr}}\in \R^{n_x \times n_x}$ and $B_{\text{tr}} \in \R^{n_x\times n_u}$ are unknown true system dynamics matrices, $x_k \in \zon_x =\langle c_{\zon_x},G_{\zon_x} \rangle \subseteq  \R^{n_x}$ and $u_k \in \zon_u = \langle c_{\zon_u},G_{\zon_u} \rangle \subseteq \R^{n_u}$ are respectively the state and the input of the system at time $k \in \mathbb{Z}_{\geq 0}$, and $w_k$ denotes the noise. 
Here, $\zon_x$ and $\zon_u$ represent the time-invariant domains of states and inputs, respectively, corresponding to inherent constraints of the problem and possibly driven by physical limitations. 

\subsection{Reachability Analysis}\label{sec:reach}
Reachability analysis computes the set of states $x_k$, which can be reached given a set of uncertain initial states $\Xx_0$ and a set of uncertain inputs $\Uu_k$. More formally, it can be defined as follows.
\begin{definition}[Exact Reachable Set]
The exact reachable set $\Rr_{N}$ after $N$ time steps subject to inputs ${u_k \in \Uu_k}$, $\forall k \in\{ 0, \dots, N-1\}$, and noise $w( \cdot) \in \zon_w$, is the set of all states trajectories starting from initial set $\Xx_0$ after $N$ steps: 
\begin{align} \label{eq:R}
        \Rr_{N} = \big\{& x(N) \in \R^{n_x} \, \big| x_{k+1} = A x_k + B u_k + w_k, \nonumber\\
       & \, x_0 \in \Xx_0,
        u_k \in \Uu_k, w_k \in \zon_w: \nonumber\\ & \forall k \in \{0,...,N{-}1\}\big\}.
\end{align}
\end{definition}

\subsection{Barrier Certificate}
\label{sec:barr}

A barrier certificate is a real-valued function defined over the system's state space, where its zero level set serves to separate the unsafe region $\Xx_u$ from all potential trajectories originating from a specified set of initial states $\Xx_0$.

\begin{definition}[Barrier Certificate~\cite{prajna2004safety}]
\label{def:bc}
Consider a system as in \eqref{eq:model-linear}. A function $\Bb: \zon_x \rightarrow \R$ is a barrier certificate for a system $\mathcal{S}$ if:
\begin{align}
\label{eq:bc1}
& \Bb(x_k) \leq 0 \quad \forall x_k \in \Xx_0, \\
\label{eq:bc2}
& \Bb(x_k)>0 \quad \forall x_k \in \Xx_u, \text { and } \\
\label{eq:bc3}
& \Bb(A x_k + B u_k +w_k) \leq \Bb(x_k) \quad \forall x_k \in \zon_x \backslash \Xx_u, \nonumber\\ 
& \quad \forall u_k \in \zon_u, \forall w_k \in \zon_w .
\end{align}
\end{definition}
Note that the above definition implies that the barrier certificate $\Bb(x_k)$ is a non-increasing function with respect to the state sequence of the system.

\subsection{Problem formulation}\label{sec:problem-formulation}

Given a set of initial states $\Xx_{0} \subset \zon_x$, a set of unsafe states $\Xx_u \subset \zon_x$, we aim to guarantee that all trajectories of $\mathcal{S}$ that start from $\Xx_{0}$ never reach $\Xx_u$ by employing only the input and noisy state data of the system $\mathcal{S}$ in \eqref{eq:model-linear} and without explicit knowledge of the system matrices. We denote this safety property by $\Psi$, and its satisfaction by $\mathcal{S}$ is written as $\mathcal{S} \models \Psi$.

We have access to past $n_T$ input-state trajectories of different lengths, denoted by $T_i+1$, for $i=1,\dots,n_T$. These trajectories are denoted as $\{ u_k\}^{T_i}_{k=0}$ and $\{ x_k\}^{T_i}_{k=0}$. For notational simplicity, in the theoretical formulation, we consider a single trajectory ($n_T=1$) of adequate length $T$ and collect all the input and noisy state data into the following matrices:
\begin{align*}
     X_{+} &= \begin{bmatrix}  x_1 & x_2 &\cdots & x_T\end{bmatrix}, \\
     X_{-} &= \begin{bmatrix}  x_0& x_1 & \cdots & x_{T-1}\end{bmatrix}, \\
      U_{-} &= \begin{bmatrix} u_0 &  u_1  & \cdots & u_{T-1} \end{bmatrix}.
\end{align*} 
We define $D_-= \begin{bmatrix} X_{-}^\top&U_{-}^\top\end{bmatrix}^\top$ and denote all available past data by $D= \begin{bmatrix} X_{+}^\top&X_{-}^\top &U_{-}^\top\end{bmatrix}^\top$. We consider the following standing assumption necessary for our data-driven approach.

\begin{assumption}\label{ass:zon-noise}
The noise $w_k$ is bounded by a known zonotope, i.e., $w_k\in \zon_w = \langle c_{\zon_w},G_{\zon_w} \rangle$ $\forall\, k\in \mathbb{Z}_{\geq 0}$, which includes the origin. 
\end{assumption}
\begin{assumption}\label{ass:rank_D}
    We assume that the data matrix $D_-$ has full row rank, i.e., $\mathrm{rank}(D_-) = n_x + n_u$.
\end{assumption}

We represent the sequence of unknown noise corresponding to the available input-state trajectories as $\{w_k\}_{k=0}^{T}$. From Assumption~\ref{ass:zon-noise}, it follows that 
\[
W_{-} = \begin{bmatrix} w_0 & \cdots & w_{T-1} \end{bmatrix} \in \mzon_w = \langle C_{\mzon_w}, G_{\mzon_w} \rangle,
\]
where $C_{\mzon_w} \in \R^{n_x \times n_T}$ and $G_{\mzon_w} \in \R^{n_x \times \gamma_\zon n_T}$. Here, $\mzon_w$ denotes the matrix zonotope resulting from the concatenation of multiple noise zonotopes $\zon_w$~\cite{Alanwar2023Datadriven}.

\section{Data-Driven Safety Verification}\label{sec:DDSV}

This section describes a data-driven safety verification process for an unknown linear system in~\eqref{eq:model-linear} subject to noise. Barrier certificates~\cite{prajna2004safety} use a model of the system to verify safety per condition \eqref{eq:bc3}. Due to the presence of noise in the data, no single model can be trusted or made to fit the data precisely to use for the search of a barrier certificate. Instead, we compute a set of models that are consistent with the observed data. Importantly, this set of models is guaranteed to include the true system model, as shown in~\cite{Alanwar2023Datadriven}, and it is represented using a matrix zonotope. The following lemma provides a systematic approach to compute this set of models. As we will show later, one can search for barrier certificates for this set of models methodically.

\begin{lemma}[{\cite[Lemma 1]{Alanwar2023Datadriven}}]\label{lem:setofAB}
Given the input-state trajectories $D$ of the unknown system~\eqref{eq:model-linear}, the matrix zonotope
\begin{align}
    \mzon_\Sigma = (X_{+} \oplus- \mzon_w) D_-^\dagger,
   \label{eq:zonoAB}
\end{align}  
with $D_-=\begin{bmatrix} X_-^\top & U_-^\top \end{bmatrix}^\top$, contains all system dynamics matrices $\begin{bmatrix}A & B \end{bmatrix}$ that are consistent with the data $D$ and the noise bound $\mzon_w$. 
\end{lemma}

We compute the reachable regions by propagating the initial set $\Xx_0$ using the set of models $\mzon_\Sigma$ as presented in the following theorem. 

\begin{theorem}[{\cite[Theorem~1]{Alanwar2023Datadriven}}]
\label{th:reach_lin}
Given input-state trajectories $D$ of the system in~\eqref{eq:model-linear}, then the reachable set 
$$\hat{\Rr}_{k+1} = \mzon_{\Sigma} (\hat{\Rr}_{k} \times \Uu_{k}  ) +  \zon_w$$
over-approximates the exact reachable set, i.e., $\hat{\Rr}_{k} \supseteq \Rr_{k}$ starting from $\hat{\Rr}_{0}=\Xx_0$. 
\end{theorem}

By converting $\mzon_{\Sigma}$ into an interval matrix, we can extract the minimum and maximum possible values of $A$ and $B$, denoted by $A^{min}, A^{max}, B^{min}$ and $B^{max}$, respectively (that is, $[A^{min}\ B^{min}] = M_{\Sigma}^{min}$ and $[A^{max}\ B^{max}] = M_{\Sigma}^{max}$). 
Note that the unknown true model parameters can be bounded as $A^{min} \leq A_{\text{tr}} \leq A^{max}$ and $B^{min} \leq B_{\text{tr}} \leq B^{max}$, where the inequalities are element-wise.
Following the above, as $A_{\text{tr}}$ and $B_{\text{tr}}$ are unknown, condition \eqref{eq:bc3} can be restated using the interval matrix conversion of $\mzon_{\Sigma}$ as:
\begin{align}
\label{eq:bc3updated}
& \Bb(A x_k + B u_k +w_k) \leq \Bb(x_k) \quad \forall x_k \in \zon_x \backslash \Xx_u, \forall u_k \in \zon_u, \nonumber\\ 
& \quad \forall w_k \in \zon_w, \forall A \in [A^{min}, A^{max}], \forall B \in [B^{min}, B^{max}].
\end{align}

While condition \eqref{eq:bc3updated} is a robust condition for safety verification, using it to find a barrier certificate becomes intractable, particularly as the dimension of the system increases due to the large number of unknown parameters. 
To address this issue, we use $\mzon_{\Sigma}$ as follows. 
$\mzon_{\Sigma} = \langle C_{\mzon_{\Sigma}}, G_{\mzon_{\Sigma}}\rangle = \langle  C_{\mzon_{\Sigma}}, \zero \rangle \oplus \langle \zero, G_{\mzon_{\Sigma}} \rangle$. 
We denote $[A_c\ B_c] = C_{\mzon_{\Sigma}}$ and $[A_G\ B_G] \in \langle \zero, G_{\mzon_{\Sigma}} \rangle$. 
The unknown system dynamics $A x_k + B u_k + w_k$ can then be written as $A_c x_k + B_c u_k + (A_G x_k + B_G u_k + w_k)$ where $(A_G x_k + B_G u_k + w_k) \in \zon_d = (\langle \zero, G_{\mzon_{\Sigma}} \rangle)(\zon_{x} \times \zon_{u}) + \zon_{w}$. 
Based on this setup, we provide the following definition for the safety verification of the system using the data-driven bounds provided above.

\begin{definition}
\label{def:bcfordata}
Consider a system as in \eqref{eq:model-linear}. A function $\Bb: \zon_x \rightarrow \R$ is a barrier certificate for this system if:
    \begin{align}
    \label{eq:bcdata1}
    & \Bb(x_k) \leq 0 \quad \forall x_k \in \Xx_0, \\
    \label{eq:bcdata2}
    & \Bb(x_k)>0 \quad \forall x_k \in \Xx_u, \text { and } \\
    \label{eq:bcdata3}
    & \Bb(A_c x_k + B_c u_k + d_k) \leq \Bb(x_k) \quad \forall x_k \in \zon_x \backslash \Xx_u,\nonumber\\ 
    & \quad \forall u_k \in \zon_u, \forall d_k \in \zon_d.
    \end{align}
\end{definition}

We now state the usefulness of Definition \ref{def:bcfordata}.

\begin{theorem}
    Consider a system $\Sys$ as in \eqref{eq:model-linear}. If there exists a function $\Bb: \zon_x \rightarrow \R$ for system $\Sys$ as in Definition \ref{def:bcfordata}, then the system is safe. 
\end{theorem}

\begin{proof}
    The proof follows from Definition \ref{def:bc}. Conditions \eqref{eq:bc1} and \eqref{eq:bc2} are directly adopted. For condition \eqref{eq:bc3}, $A_c x_k + B_c u_k$ captures the nominal trajectory estimated by the matrix zonotope over-approximation, and $d_k$ captures the ``noise'' from the over-approximation. The term $A_c x_k + B_c u_k + d_k$ thus captures all possible trajectories generated from the matrix zonotope $\mzon_{\Sigma}$ starting from state $x_k$ under input $u_k$. Thus, condition \eqref{eq:bcdata3} replaces condition \eqref{eq:bc3} where the true trajectory $Ax_k + Bu_k + w_k$ is captured in $A_c x_k + B_c u_k + d_k$.
\end{proof}


The next section covers a computational method to find barrier certificates that satisfy conditions \eqref{eq:bcdata1}-\eqref{eq:bcdata3} for the data-driven set of models.

\section{Computation of Safety Certificates}
\label{sec:computation}
This section presents sum-of-squares (SOS) programming as a computational method of searching for barrier certificates for safety verification of our data-driven over-estimated system. 
To find barrier certificates, we first fix the template to be a linear combination of user-defined basis functions:
\begin{align*}
   \Bb(x) = \mathbf{c}^T\mathbf{p}(x) = \sum_{i = 1}^{n} c_i p_i(x), 
\end{align*}
where functions $p_i$ are monomials over the state variable $x$, and $c_1, \ldots, c_n$ are the real coefficients.

The system dynamics \eqref{eq:model-linear} is linear.
We say a set $Q \subseteq \R^n$ is semi-algebraic if it can be defined with the help of a vector of polynomial inequalities $h(x)$ as $Q = \{ x \mid h(x) \geq 0 \}$, where the inequalities are element-wise.
When the initial set $\Xx_0$ and unsafe set $\Xx_u$ are semi-algebraic~\cite{bochnak2013real}, conditions \eqref{eq:bcdata1}-\eqref{eq:bcdata3} can be cast as a collection of SOS constraints in order to compute a polynomial barrier certificate of a predefined degree. We note that (matrix) zonotopes are semi-algebraic sets. As the conversion of (matrix) zonotopes to their halfspace representation can be exponential in the number of generators~\cite{althoff2010computing}, we instead convert them to interval matrices as they are scalable and computationally convenient due to their ease of representation as polynomial inequalities.
\begin{assumption}
    \label{asp:sos}
    The state set $\zon_x$ is a subset of $\mathbb{R}^n$, and the system dynamics \eqref{eq:model-linear} is a linear function of the state $x$ and input $u$.
    Furthermore, sets $\zon_x$, $\Xx_0$, $\Xx_u$, $\zon_u$ and $\zon_d$ are zonotopes and their interval matrix representations can be described as vectors of polynomial inequalities: $\zon_{x,int} = \{x\in \mathbb{R}^n \mid g(x) \geq 0\},\ \Xx_{0,int} = \{x\in \mathbb{R}^n \mid g_0(x) \geq 0\}$, $\Xx_{u,int} = \{x\in \mathbb{R}^n \mid g_u(x) \geq 0\}$, $\zon_{u,int} = \{u\in \mathbb{R}^n \mid g_{in}(u) \geq 0\}$, and $\zon_{d,int} = \{d\in \mathbb{R}^n \mid g_d(d) \geq 0\}$, where $g(\cdot), g_0(\cdot), g_u(\cdot), g_{in}(\cdot)$, and $g_d(\cdot)$ are each a vector of polynomials of the form in \eqref{eq:poly_inequality} and the inequalities are element-wise.
\end{assumption}

Under Assumption \ref{asp:sos}, conditions \eqref{eq:bcdata1}-\eqref{eq:bcdata3} can be formulated as a set of SOS constraints, as follows.

\begin{lemma}
    \label{lem:sos}
    Consider a system given by equation \eqref{eq:model-linear}. Suppose Assumption \ref{asp:sos} holds for this system and there exist constants $\epsilon \in \mathbb{R}_{>0}$, polynomial $\mathcal{B}(x)$ and SOS polynomials $\lambda_0(x),\ \lambda_{u}(x),\ \lambda(x)$ of appropriate dimensions such that:
    \begin{align}
        \label{eq:sos_init}
        &-\Bb(x) - \lambda_0^T(x) g_0(x),\\
        \label{eq:sos_unsf}
        &\Bb(x) - \epsilon - \lambda_{u}^T(x) g_u(x),\\
        \label{eq:sos_dyn}
        &\Bb(x) - \mathcal{B}(A_c x + B_c u + d) - \lambda^T(x,u,d) g(x,u,d),
    \end{align}
    are SOS polynomials where $x$ is the state variable over $\zon_{x,int}$, $u$ is the input variable over $\zon_{u,int}$, $d$ is the ``noise'' variable over $\zon_{d,int}$ and $A_c, B_c$ are the data-driven nominal model parameters of the system. Here, $g(x,u,d)$ is the concatenation of $g(x),g_{in}(u)$, and $g_d(d)$.
    Then the function $\Bb(x)$ is a barrier certificate following Definition \ref{def:bcfordata}. Note that $\epsilon$ is introduced in condition \eqref{eq:sos_unsf} to convert the strict inequality in condition \eqref{eq:bcdata2} to an inclusive inequality.
\end{lemma}

In the next section, we demonstrate the application of our proposed approach over a case study.

\section{Numerical simulations}
\label{sec:numerical-simulations}
In this section, we demonstrate the application of our approach in a case study. The simulations were conducted on a Windows 11 device equipped with an AMD Ryzen 9 4900HS Mobile Processor and 16\,GB of RAM. 
We used CORA v2020~\cite{Althoff2015ARCH} in MATLAB for the matrix zonotope related computations and transferred the resulting model data to Julia where we used TSSOS~\cite{wang2021tssos} to implement the SOS optimization problem.
We first obtain $\mzon_\Sigma$ and attempt to search for a barrier certificate using conditions \eqref{eq:bcdata1}, \eqref{eq:bcdata2}, and \eqref{eq:bc3updated}. For state $x \in \R^{n_x}, u \in \R^{n_u}$, the total number of parameters according to these conditions is $n_x \times n_x + n_x + n_x \times n_u + n_u + n_x = n_x^2 + n_xn_u + 2n_x + n_u$. This does not scale favorably with respect to the state dimension $n_x$ and input dimension $n_u$. 
With each additional degree of the polynomial template for the barrier certificate, the memory needed scales exponentially over these parameters.
Thus, it was not computationally tractable for our case study beyond a two dimensional system. 
Note that condition \eqref{eq:sos_dyn} in Lemma \ref{lem:sos} is replaced with 
\begin{align*}
    &\Bb(x) - \mathcal{B}(Ax + Bu + w) \\ 
    &\quad - \lambda^T(x,u,w,B,A) g(x,u,w,B,A),
\end{align*}
to incorporate the semi-algebraic set expressions for $w,A$ and $B$ for the optimization process based on condition \eqref{eq:bc3updated}.
Figure \ref{fig:hard2d} shows the resulting sublevel set of the barrier certificate for a two dimensional system along with the designated initial and unsafe states. Table \ref{tab:hard2d} displays the true dynamics used for simulation (data collection) as well as the relevant vectors and matrices used for the SOS optimization.

\begin{figure}[!t]
    \centering    \epsfig{file=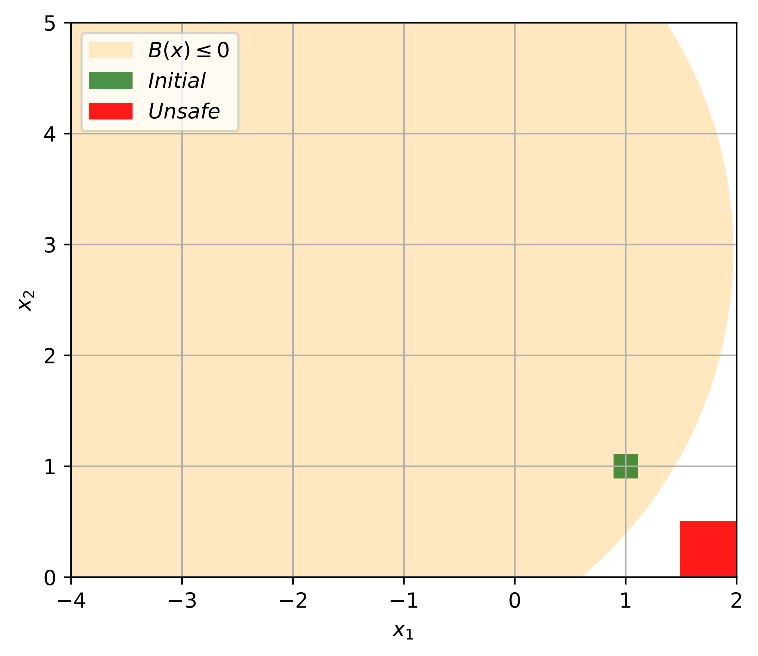,width=0.45\textwidth, keepaspectratio}
    \caption{Barrier Certificate for a 2D System.} 
    \label{fig:hard2d}
\end{figure}

\begin{figure}[!t]
    \centering
  \begin{subfigure}[b]{0.5\textwidth}
    \epsfig{file=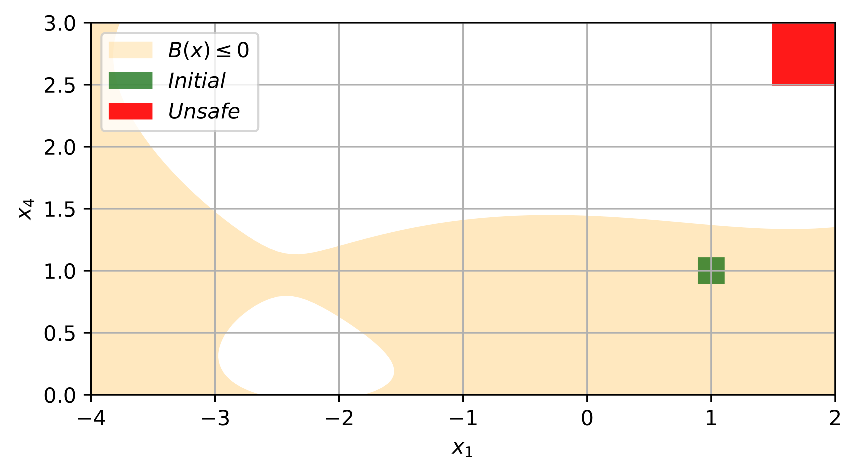, width=\textwidth, keepaspectratio}
  \end{subfigure}\\
  \begin{subfigure}[b]{0.5\textwidth}
    \epsfig{file=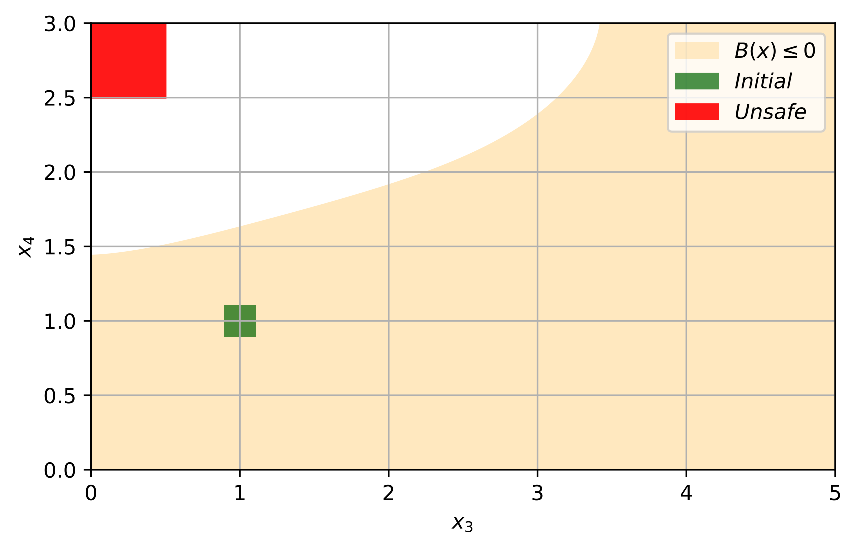, width=\textwidth, keepaspectratio}
  \end{subfigure}
  \caption{Barrier Certificate Projections for a 5D System.}
  \label{fig:relaxed5d}
\end{figure}

We now present our results for a five dimensional system following Definition \ref{def:bcfordata} and using Lemma \ref{lem:sos}.
We were able to find a degree five polynomial barrier certificate. Two examples of the projection of the certificate over two states of the system are shown in Figure \ref{fig:relaxed5d}. Table \ref{tab:relaxed5d} displays the true dynamics used for simulation as well as the relevant vectors and matrices used for the SOS optimization.

\begin{table}[!b]
    \centering
    \caption{Relevant Information for the 2D System}
    \begin{tabular}{|l|l|}
    \hline
     $\zon_{x,int}$ & $[-4,0]^T \leq x \leq [2,5]^T$\\ \hline
     $\zon_{u,int}$ & $9.75 \leq u \leq 10.25$ \\ \hline
     $\zon_{w,int}$ & $-0.005\one_2 \leq w \leq 0.005\one_2$ \\ \hline
     $\Xx_0$ & $0.9\one_2 \leq x \leq 1.1\one_2$ \\ \hline
     $\Xx_u$ & $[1.5,0]^T \leq x \leq [2,2.5]^T$ \\ \hline
     $[A^{min}\ B^{min}]$ & 
     $\begin{matrix}
         0.928 & -0.194 & 0.0418\\
         0.185 & 0.927 & 0.0514
     \end{matrix}$\\ \hline
     $[A^{max}\ B^{max}]$ & 
     $\begin{matrix}
         0.935 & -0.184 & 0.0452\\
         0.192 & 0.937 & 0.0549
     \end{matrix}$ \\ \hline
     $[A_{\text{tr}}\ B_{\text{tr}}]$ & 
     $\begin{matrix}
         0.932 & -0.189 & 0.0436\\
         0.189 & 0.932 & 0.0533
     \end{matrix}$\\ \hline
    \end{tabular}
    \label{tab:hard2d}
\end{table}

\begin{table}[!t]
    \centering
    \caption{Relevant Information for the 5D System}
    \begin{tabular}{|l|l|}
    \hline
     $\zon_{x}^{min}$ & $[-4,0,0,0,0]^T$\\ \hline
     $\zon_{x}^{max}$ & $[2,5,5,3,6]^T$\\ \hline
     $\zon_{u,int}$ & $9.75 \leq u \leq 10.25$ \\ \hline
     $\zon_{d,int}$ & $-2.173\one_5 \leq u \leq 2.173\one_5$ \\ \hline
     $\Xx_0$ & $0.9\one_5 \leq x \leq 1.1\one_5$ \\ \hline
     $\Xx_u^{min}$ & $[1.5,0,0,2.5,0]^T$ \\ \hline
     $\Xx_u^{max}$ & $[2,0.5,0.5,3,0.5]^T$ \\ \hline
     $A^{min}$ & 
     $\begin{matrix}
         0.921 & -0.20 & -0.193 & -0.168 & -0.0877 \\
         0.178 & 0.922 & -0.193 & -0.168 & -0.0877 \\ 
         -0.011 & -0.0106 & 0.667 & -0.125 & -0.0877 \\ 
         -0.011 & -0.0106 & -0.236 & 0.692 & -0.0877 \\ 
         -0.011 & -0.0106 & -0.193 & -0.168 & 0.817 \\ 
     \end{matrix}$\\ \hline
     $B^{min}$ & $[0.0357,0.0454,0.0397,0.0374,0.0397]^T$\\ \hline
     $A^{max}$ & 
     $\begin{matrix}
     0.941 & -0.177 & 0.171 & 0.18 & 0.0989\\
       0.198 & 0.944 & 0.171 & 0.18 & 0.0989 \\ 
       0.0086 & 0.0118 & 1.031 & 0.223 & 0.0989 \\ 
       0.0086 & 0.0118 & 0.128 & 1.04 & 0.0989 \\ 
       0.0086 & 0.0118 & 0.171 & 0.18 & 1.004 \\
     \end{matrix}$ \\ \hline
     $B^{max}$ & $[0.0516,0.0612,0.0555,0.0532,0.0555]^T$\\ \hline
     $A_{c}$ & 
     $\begin{matrix}
     0.931 & -0.188 & -0.0109 & 0.00579 & 0.00557 \\
        0.188 & 0.933 & -0.0109 & 0.0058 & 0.0056 \\ 
        -0.0012 & 0.00061 & 0.849 & 0.0488 & 0.00557 \\ 
        -0.0012 & 0.00061 & -0.0539 & 0.865 & 0.00557 \\ 
        -0.0012 & 0.00061 & -0.0109 & 0.00579 & 0.91 \\
     \end{matrix}$ \\ \hline
     $B_{c}$ & $[0.0437,0.0533,0.0476,0.0453,0.0476]^T$\\ \hline
     $A_{\text{tr}}$ & 
     $\begin{matrix}
        0.1890 & 0.9323 & 0 & 0 & 0 \\
         0 & 0 & 0.8596 & 0.0430 & 0 \\
         0 & 0 & -0.0430 & 0.8596 & 0 \\
         0 & 0 & 0 & 0 & 0.9048 \\
     \end{matrix}$ \\ \hline
     $B_{\text{tr}}$ & $[0.0436,0.0533,0.0475,0.0453,0.0476]^T$\\ \hline
    \end{tabular}
    \label{tab:relaxed5d}
\end{table}

\section{Conclusion} \label{sec:conclusion}

This work presents a data-driven framework for safety verification of cyber-physical systems, addressing challenges posed by model uncertainty, noise, and environmental disturbances. By leveraging matrix zonotopes to construct a set of possible system models, our approach ensures that all dynamics consistent with the observed data are accounted for, providing a robust alternative to traditional model-based methods. The integration of this representation with barrier certificates enables rigorous safety guarantees without requiring explicit knowledge of the system model. Numerical experiments validate the effectiveness of our method in verifying safety for dynamical systems with unknown models. We plan to adopt our method for enhanced model parameters and generalize it for polynomial nonlinear control systems. Furthermore, we intend to investigate safe controller synthesis for a given set of models using barrier certificates by building upon this work.

\bibliographystyle{alpha}
\bibliography{ref}

\end{document}